\documentclass[submission,copyright,creativecommons]{eptcs}
\providecommand{\event}{NCMA 2022} % Name of the event you are submitting to
\usepackage{breakurl}             % Not needed if you use pdflatex only.
\usepackage{underscore}           % Only needed if you use pdflatex.

\usepackage{amsmath}
\usepackage{amsthm}
\usepackage{amssymb}
\usepackage{multicol}

\usepackage{thmtools}

\newcommand{\resp}{respectively}
\newcommand{\dfa}{\text{DFA}}
\newcommand{\pfa}{\text{PFA}}

\newcommand{\N}{\ensuremath{\mathbb{N}}}

\newcommand{\eps}{\varepsilon}
\newcommand{\dsc}{\sc}

\renewcommand{\sc}{\mathrm{sc}}

\newcommand{\asc}{\mathrm{asc}}

\newtheorem{thm}{Theorem}
\newtheorem{lem}[thm]{Lemma}
\newtheorem{con}[thm]{Conjecture}
\newtheorem{cor}[thm]{Corollary}
\newtheorem{exa}[thm]{Example}

\title{On the Accepting State Complexity of Operations on Permutation Automata}

\author{Christian Rauch\qquad\qquad Markus Holzer 
\institute{Institut f\"ur Informatik, Universit\"at Giessen,
  Arndstr.~2, 35392 Giessen, Germany}
\email{\quad christian.rauch@informatik.uni-giessen.de\quad\qquad
  holzer@informatik.uni-giessen.de} }

\def\titlerunning{Accepting State Complexity of Operations on
  Permutation Automata}

\def\authorrunning{\textsc{C.~Rauch}, \textsc{M.~Holzer}}

\begin{document}
 
\maketitle

\begin{abstract}
  We investigate the accepting state complexity of deterministic
  finite automata for regular languages obtained by applying one of
  the following operations to languages accepted by permutation
  automata: union, quotient, complement, difference, intersection,
  Kleene star, Kleene plus, and reversal. The paper thus joins the
  study of accepting state complexity of regularity preserving
  language operations which was initiated by the work [J. Dassow: On
  the number of accepting states of finite automata, J.~Autom.,
  Lang. Comb., 21, 2016].  We show that for almost all of the
  operations, except for reversal and quotient, there is no difference
  in the accepting state complexity for permutation automata compared
  to deterministic finite automata in general. For both reversal and
  quotient we prove that certain accepting state complexities cannot
  be obtained; these number are called ``magic'' in the
  literature. Moreover, we solve the left open accepting state
  complexity problem for the intersection of unary languages accepted
  by permutation automata and deterministic finite automata in
  general.

  % \keywords permutation automata, language operations, accepting
  % state complexity, descriptional complexity.
\end{abstract}

\section{Introduction}\label{sec:in}

The state complexity of a regular language is a classical
well-understood descriptional complexity measure of finite state
systems, that is defined to be the number of states of the smallest,
either deterministic or nondeterministic, finite automaton that
recognizes it. It has been studied from different perspectives in the
literature like, for instance, (i) for regular languages in general
and for certain sub-families, (ii) for converting nondeterministic
finite automata to equivalent deterministic finite automata, and (iii)
for operations, called the operational complexity, on regular
languages in general and sub-families thereof. For a brief survey on
the subject we refer to, e.g.,~\cite{GMRY16}.

Recently, the accepting state complexity of a regular language was
introduced in~\cite{Da16}. It is defined to be the minimal number of
accepting states needed for a finite state device, either
deterministic or nondeterministic, that accepts it. While the
accepting state complexity forms a strict hierarchy of language
classes for deterministic finite automata, it collapses for
nondeterministic state devices, since every regular language not
containing the empty word is accepted by a nondeterministic finite
automaton with a single final state. If the empty word belongs to the
language, the nondeterministic accepting state complexity is at most
two. Thus, the conversion from nondeterministic to equivalent
deterministic finite automata can produce unbounded deterministic
accepting state complexity for a regular language. Moreover, the
operational accepting state complexity was studied
in~\cite{HoHo18}. The obtained results on the accepting state
complexity prove that this measure is significantly different to the
original state complexity. What is missing for the accepting state
complexity is a study for certain sub-families of the regular
languages in order to better understand the intrinsic behaviour of
this measure.

We close this gap by studying the operational accepting state
complexity for the class of permutation automata (\pfa s) which accept
the so called p-regular languages, also named pure-group
languages. This language family is of particular interest from an
algebraic point of view since their syntactic monoid induces a
group. Additionally permutation automata together with
permutation-reset automata play a key role in the decomposition of
deterministic finite automata (\dfa s), see, e.g.,~\cite{Ze67}.  It is
also worth to mention that the class of p-regular languages was one of
the first subclasses of the regular languages for which the star
height problem was shown to be decidable, see,
e.g.,~\cite{Br80}. Recently, the family of p-regular languages, and
thus \pfa s, gained renewed interest.  For instance, in~\cite{JMW21}
the decomposing of \pfa s into the intersection of smaller automata of
the same kind was investigated.  Moreover the operational state
complexity on \pfa s was studied in~\cite{HoMl20}. Up to our knowledge
the operational accepting state complexity of p-regular languages was
not investigated so far.  We study this problem by examining the
following question:
\begin{quote}
  Given are three non-negative integers~$m$,\ $n$, and $\alpha$ and a
  regularity preserving language operation~$\circ$, are there minimal
  permutation automata~$A$ and~$B$ with accepting state complexity~$m$
  and~$n$, \resp, such that the language~$L(A)\circ L(B)$ is accepted
  by a minimal deterministic finite automaton with~$\alpha$ accepting
  states?
\end{quote}
Following the terminology of~\cite{IKT00} we call values~$\alpha$
``magic'' if there are no such automata~$A$ and~$B$. The following
results were shown in~\cite{Da16} and~\cite{HoHo18} for the
operational accepting state complexity on languages accepted by \dfa
s---for accepting state complexities~$m$ and~$n$ one can obtain all
values from the given number set for~$\alpha$:
%\begin{multicols}{2}
  \begin{itemize}
  \item Complement: $\N\cup\{\,0 \mid m=1\,\}$.
  \item Kleene star and Kleene plus: $\N$.
  \item Union: $\N$.
  \item Set difference: $\N$.
  \item Intersection: $[0,mn]$ (if the input alphabet is at least
    binary).
  \item Reversal: $\N$.
  \item Quotient: $\N\cup\{ 0\}$.
  \end{itemize}
%\end{multicols}
One may have noticed that for none of the above mentioned operations
magic numbers exist up to the special case of complementation
and~\mbox{$m=1$}.\footnote{The intersection of two languages~$L_1$
  and~$L_2$ of accepting state complexity~$m$ and~$n$ is accepted by
  the cross product of the minimal \dfa s accepting~$L_1$
  and~$L_2$. So the accepting state complexity is directly bounded
  by~$mn$. Therefore the numbers greater~$mn$ are not of interest.}
It is worth mentioning that these results were obtained without the
use of \pfa s automata witnesses.

We generalize these results for the operations complement, union, set
difference, Kleene star, and Kleene plus to the family of languages
accepted by \pfa s. This means even though \pfa s are restricted in
their expressive power compared to arbitrary \dfa s, there are no
magic numbers for the accepting state complexity of \pfa s w.r.t.\ the
above mentioned operations.  When considering the reversal operation a
significant difference appears. While for \dfa s the reversal
operation induces the whole set~$\N$ as accepting state complexities
as mentioned above, this is not the case for \pfa s, where we can
prove that the number $\alpha=1$ is magic for every $m\geq 2$. In
fact, we prove that for $m=2$ no other magic number as $\alpha=1$
exists. Whether this is also true for larger~$m$ is left open. Yet
another difference in the accepting state complexity comes from the
quotient operation. Here it turns out that for unary languages
accepted by \pfa s only the range $[1,mn]$ is obtainable for the
accepting state complexity. This is entirely different compared to the
general case.  Finally, the unary case for the accepting state
complexity of the intersection operation for \dfa s in general was
left open in~\cite{HoHo18}. We close this gap by considering this
problem in detail. In this way, we identify a whole range of magic
numbers for the intersection of unary languages accepted by \pfa s and
extend this result to the case of \dfa s, solving the left open
problem mentioned above.  Due to space constraints some of the proofs
are omitted; they can be found in the full version of this paper.

\section{Preliminaries} \label{sec:defs}

Let~$\N$ denote the set of all positive integers and~$\N_{\geq x}$
($\N_{\leq x}$, \resp) the set of all positive integers that are
greater or equal~$x$ (less or equal~$x$, \resp).

We recall some definitions on finite automata as contained
in~\cite{Ha78}.  Let~$\Sigma^*$ denote the set of all words over the
finite alphabet~$\Sigma$.  The \emph{empty word} is denoted by~$\eps$.
Further, we denote the set $\{i,i+1,\ldots,j\}$ by $[i,j]$, if~$i$
and~$j$ are integers.  A \emph{deterministic finite automaton} (\dfa)
is a quintuple $A=(Q,\Sigma,{}\cdot{},s,F)$, where~$Q$ is the finite
set of \textit{states}, $\Sigma$ is the finite set of \textit{input
  symbols}, $s\in Q$ is the \textit{initial state}, $F\subseteq Q$ is
the set of \textit{accepting states}, and the \textit{transition
  function}~$\cdot$ maps~$Q\times\Sigma$ to~$Q$.  The \emph{language
  accepted} by the \dfa~$A$ is defined as
$$L(A) =\{\,w\in \Sigma^*\mid\mbox{$s\cdot w\in F$}\,\},$$ 
where the transition function is recursively extended to a map
$Q\times\Sigma^*$ onto~$Q$. Obviously, every letter $a\in\Sigma$
induces a mapping on the state set~$Q$ to~$Q$ by $q\mapsto q\cdot a$,
for every $q\in Q$. A \dfa\ is \textit{unary}, if the input
alphabet~$\Sigma$ is a singleton set, that is, ~$\Sigma = \{a\}$, for
some input symbol~$a$. Moreover, if every letter of the automaton
induces only permutations on the state set, then we simply speak of a
\emph{permutation} automaton~(\pfa).

As usual we denote the \emph{state complexity} of a language~$L$
accepted by a \dfa\ by
$$\dsc(L)=\min\{\,\dsc(A)\mid\mbox{$A$ is a \dfa\ with $L=L(A)$}\,\},$$
where $\dsc(A)$ refers to the number of states of the
automaton~$A$. Similarly we define the measure $\asc(L)$ the
\emph{accepting state complexity} of a language~$L$ accepted by a
\dfa, where $\asc(A)$ refers to the number of final states of the
automaton~$A$.
 
An automaton is \emph{minimal} (a-minimal, \resp) if it admits no
smaller equivalent automaton w.r.t.\ the number of states (final
states, \resp). For \dfa s both properties can be easily
verified. Minimality can be shown if all states are reachable from the
initial state and all states are pairwise inequivalent. For
a-minimality the following result shown in~\cite{Da16} applies:

\begin{thm}
  Let~$L$ be a language accepted by a minimal \dfa~$A$. Then
  $\asc(L)=\asc(A)$.
\end{thm}

In order to characterize the behaviour of complexities under
operations we introduce the following notation: for
$c\in\{\dsc,\asc\}$, a $k$-ary regularity preserving operation~$\circ$
on languages, and natural numbers $n_1,n_2,\ldots,n_k$, we define
$$g^c_\circ(n_1,n_2,\ldots,n_k)$$ 
as the set of all integers~$r$ such that there are~$k$ regular
languages $L_1,L_2,\ldots,L_k$ with $c(L_i)=n_i$, for $1\leq i\leq k$,
and $c(\circ(L_1,L_2,\ldots,L_k))=r$. In case we only consider unary
languages $L_1,L_2,\ldots, L_k$ we simply write~$g^{c,u}_\circ$
instead. When restricting the underlying languages to, e.g., be
accepted by permutation automata (\pfa s), we indicate this by writing
$g^c_{\circ,\pfa}$ and $g^{c,u}_{\circ,\pfa}$, \resp.

In order to explain the notation we give a small example.

\begin{exa}\label{exa:asc-for-complementation}
  Consider the unary operation~$C$ of complementation of languages. It
  is obvious that
$$g^{\dsc}_C(m)=\{m\},\quad\mbox{for $m\geq 1$.}$$
On the other hand, when we consider the accepting state complexity,
in~\cite{Da16} the following behaviour
$$
g^{\asc}_C(m)=
\begin{cases}
  \{1\} & \mbox{if $m=0$,}\\
  \{0\}\cup\mathbb{N} & \mbox{if $m=1$,}\\
  \mathbb{N} & \mbox{otherwise,}
\end{cases}
$$
for the complementation was proven. Moreover, it is easy to see that
$$g^{\sc,u}_C(m)=g^{\sc}_C(m)\quad\mbox{and}\quad g^{\asc,u}_C(m)=g^{\asc}_C(m)$$
holds.\qed
\end{exa}

In the constructions to come, note that we will use the~$\bmod$
operation such that~$x\bmod y +z$ is the same as~$(x\bmod y) +z$ and
not equal to~$x\bmod (y +z)$, but $x+y\bmod z$ is the same as
$(x+y)\bmod z$.  We use~$\div$ for the integer division and~$/$ for the
ordinary division.

\section{Results}
\label{sec:results}

We investigate the accepting state complexity of various regularity
preserving language operations such as union, quotient, complement,
difference, intersection, Kleene star, Kleene plus, and reversal on
languages accepted by permutation automata. Before we start our
investigation we introduce a useful notion for \emph{unary}
permutation automata by strings. Since a unary permutation automaton
consists of a cycle only, it suffices to encode the finality of these
cycle states by a binary string. This is done as follows: a word
$w\in\{0,1\}^+$ with $w=a_0a_1\ldots a_{|w|-1}$, for $a_i\in\{0,1\}$
and $0\leq i\leq |w|-1$, describes the permutation automaton
$$A_w=(\{0,1,\ldots,|w|-1\},\{a\},{}\cdot{},0,\{\,i\mid\mbox{$0\leq i<|w|-1$ and $a_i=1$}\,\})$$ with
$$
i\cdot a=
\begin{cases}
  i + 1 & \mbox{for $0\leq i<|w|-1$,}\\
  0 & \mbox{otherwise.}
\end{cases}
$$
It is clear that there is a bijection between words in $\{0,1\}^+$
with all unary \pfa s. Thus, we can identify words with \pfa s and
\textit{vice versa}.  Now we are ready for the investigation of the
accepting state complexity of certain operations on \pfa s.

\subsection{Complementation}

The complement of a language accepted by a finite automaton can be
obtained by simply interchanging accepting and non-accepting
states. Hence, the state complexity of a language accepted by a finite
automaton, regardless whether the automaton is a permutation automaton
or not, is the same. The result on the \emph{accepting state
  complexity} for unrestricted \dfa s was presented in
Example~\ref{exa:asc-for-complementation}. Next we show that this
result even holds for \pfa s.

\begin{thm}\label{thm:complement}
  We have
  $g^{\asc,u}_{C,\pfa}(m)=g^\asc_{C,\pfa}(m)=g^{\asc}_C(m)=g^{\asc,u}_C(m)$.
\end{thm}

\subsection{Kleene Star and Kleene Plus}

Next we study the accepting state complexity of the Kleene star and
the Kleene plus operations for permutation automata. We want to
mention that the Kleene closure of a p-regular language cannot be
accepted by a \pfa\ in many cases, see for example~\cite{HoMl20}.
First we prove a useful relation between \pfa s and the languages
which they accept.

\begin{lem}\label{lem:pfa-language}
  Let~$I$ be a finite set of non-negative integers and~$j$ be a number
  which is greater than the biggest number in~$I$. The
  language~$\bigcup\limits_{i\in I} a^{i}(a^{j})^*$ can be accepted by
  the \pfa
  \[A=(\{q_0,q_1,\ldots, q_{j-1}\},\{a\},{}\cdot_A{},q_0,\{\,q_i\mid
  i\in I\,\})\] with~$q_i\cdot_A a=q_{i+1\bmod j}$.  Additionally~$A$
  is minimal if there is no divisor~$t$ of~$j$ such that for
  every~$i\in I$ the number~$i+t\bmod j$ is in~$I$.
\end{lem}

\begin{proof}
  The tedious details for the first statement are left to the reader.
  We prove the second statement by contradiction. So assume that there
  is a divisor~$t$ of~$j$ such that for every~$i\in I$ the
  number~$i+t\bmod j$ is in~$I$.  Assuming that~$A$ is minimal, for
  every pair of states there is a word~$w$ which distinguishes them,
  i.e., maps one of the states onto an accepting and the other one
  onto a non-accepting state. This includes the states~$q_i$
  and~$q_{i+t\bmod j}$.  Since~$w\in a^*$ we can assume~$w=a^k$ for a
  non-negative integer~$k$. Due to the definition of~$A$ the word~$w$
  maps the states~$q_i$ and~$q_{i+t\bmod j}$ onto~$q_{i+k\bmod j}$
  and~$q_{i+t+k\bmod j}$ which are either both in~$I$ or are both not
  in~$I$.  This contradicts the assumption that~$w$
  distinguishes~$q_i$ and~$q_{i+t\bmod j}$.
\end{proof}

We will use this result to prove that no magic numbers exist for the
accepting state complexity of the Kleene star operation.

\begin{thm}\label{thm:star}
  We have
	$$g_{*}^{\asc,u}(m)=g_{*,\pfa}^{\asc,u}(m)=
	\begin{cases}
          \{1\} & \mbox{if $m=0$,}\\
          \N & \mbox{otherwise.}
	\end{cases}
	$$
      \end{thm}

\begin{proof}
  For~$m=0$ we observe that~$\emptyset^*=\{\eps\}$, for~$\eps$
  being the empty word. So the first statement follows. For the second
  one we distinguish whether~$\alpha$ or~$m$ are equal to one. We
  distinguish four cases, where in each case we use
  Lemma~\ref{lem:pfa-language} to show that the defined language has
  accepting state complexity~$m$:
  \begin{enumerate}
  \item Case~$\alpha=1$ and~$m>1$: The
    language~$L=\bigcup\limits_{i=1}^m a^{i}(a^{m+1})^*$ has accepting
    state complexity~$m$ and its Kleene star is equal to~$\Sigma^*$
    which has accepting state complexity one.
			
  \item Case~$\alpha=1$ and~$m=1$: The language of the previous case
    can also be used if~$m=1$.
			
  \item Case~$\alpha>1$ and~$m>1$:
    Let~$L=a^2(a^{2(\alpha-1)+m+1})^*\cup \bigcup\limits_{i=1}^{m-1}
    a^{2(\alpha-1)+i}(a^{2(\alpha-1)+m+1})^*$. Then~$L^*$ is equal to
    $\bigcup\limits_{i=0}^{\alpha-2} a^{2i}\cup
    a^{2(\alpha-1)}\Sigma^*$. In turn~$\bigcup\limits_{i=0}^{\alpha-2}
    a^{2i}\cup a^{2(\alpha-1)}\Sigma^*$ can be accepted by a unary
    \dfa\ which has a tail of length~$2(\alpha-1)$ and a cycle formed
    by one state, where all states on positions with an even number
    are accepting if we start counting by zero.
			
  \item Case~$\alpha>1$ and~$m=1$: Define
    $L=a^2(a^{2\alpha-1})^*$. Then $L^*$ is equal to
    $\bigcup\limits_{i=0}^{\alpha-2} a^{2i}\cup
    a^{2(\alpha-1)}\Sigma^*$, which is the Kleene star language from
    the previous case.
  \end{enumerate}
  So in all cases the Kleene closure has accepting state
  complexity~$\alpha$ which completes the proof.
\end{proof}

By taking into account that for every language~$L$ the empty
word~$\eps$ is in~$L^+$ if and only if~$\eps\in L$, with a
small adjustment of the used languages for the previous theorem we
obtain the following corollary for the Kleene plus operation.

\begin{cor}\label{cor:plus}
  We have
  $$g_{+}^{\asc,u}(m)=g_{+,\pfa}^{\asc,u}(m)=
	\begin{cases}
          \{0\} & \mbox{if $m=0$,}\\
          \N & \mbox{otherwise.}
	\end{cases}
	$$
\end{cor}

 \subsection{Union}

 In this subsection we extent the results for the accepting state
 complexity from~\cite{Da16} for the union operation to the class of
 permutation automata. For \dfa s in general the following result was
 shown in~\cite{Da16}:
$$
g^{\asc,u}_{\cup}(m,n)=g^{\asc}_{\cup}(m,n)=
\begin{cases}
  \{m\} & \mbox{if $n=0$,}\\
  \{n\} & \mbox{if $m=0$,}\\
  \N & \mbox{otherwise,}
\end{cases}
$$
and since the union operation is commutative
$g^{\asc,u}_{\cup}(m,n)=g^{\asc,u}_{\cup}(n,m)$ and
$g^{\asc}_{\cup}(m,n)=g^{\asc}_{\cup}(n,m)$. Note the languages that
prove these results are \emph{not} accepted by any \pfa.

We will prove that except for the special cases~$m=0$ or~$n=0$ every
accepting state complexity can be reached also for unary
alphabets. Therefore the reachable numbers coincide in the cases that
the input \dfa s are restricted or not. We split this into three
theorems which show that~$\N_{\leq \min\{n,m\}},$ $\N_{\geq
  \max\{n,m\}}$ and~$[\min\{n,m\}+1,\max\{n,m\}-1]$ are reachable.  We
start with the upper range.

\begin{thm}\label{thm:union-upper-range}
  For~$m,n\geq 1$ we have $\N_{\geq \max\{n,m\}}\subset
  g_{\cup,\pfa}^{asc,u}(m,n)$.
\end{thm}

We split the proof of this theorem into two lemmata, which show the
reachability of smaller intervals (Lemma \ref{lem:union-upper-range})
and that the union of those intervals equals the whole range~$\N_{\geq
  \max\{n,m\}}$ (Lemma \ref{lem:union-upper-range-union}).

\begin{lem}\label{lem:union-upper-range}
  Let~$m\geq n\geq 1$, $i\geq 1$ and~$\alpha\in
  [\max\{in,m\},in+m]$. There are minimal unary \pfa s~$A$ and~$B$
  with accepting state complexity~$m$ and~$n$, \resp, such
  that~$L(A)\cup L(B)$ has accepting state complexity~$\alpha$.
\end{lem}

The next lemma shows that the union of the intervals which are
reachable due to the previous lemma is again an interval.

\begin{lem}\label{lem:union-upper-range-union}
  For~$m\geq n$ holds $\bigcup_{i \in \N} [\max\{in,m\},in+m]=\N_{\geq
    m}$.
\end{lem}

It is not hard to see that the Lemmata~\ref{lem:union-upper-range}
and~\ref{lem:union-upper-range-union} are symmetric in~$n$ and~$m$ so
together they prove Theorem~\ref{thm:union-upper-range}. Next we show
that the lower range is reachable, too.

\begin{thm}
  \label{thm:union-lower-range}
  For~$m,n\geq 1$ we have $[1,\min\{n,m\}]\subset
  g_{\cup,\pfa}^{asc,u}(m,n)$.
\end{thm}

The constructions for the previous lemmata created
sequences~$1^\alpha0^\ell$ for some~$\ell\in\N$. There are values
for~$n,$ $m$ and~$\alpha$ such that~$\alpha$ cannot be reached by this
method. Therefore we have to create sequences which contain~$\alpha$
accepting and distinguishable states which are not consecutive.  We
want to mention here that we count the positions of the states in an
unary \pfa\ in the same way we count the positions of the letters
describing the \pfa, namely we start by zero.

\begin{thm}
  \label{thm:union-mid-range}
  For~$m,n\geq 1$ we have $[\min\{n,m\}+1,\max\{n,m\}-1]\subset
  g_{\cup,\pfa}^{asc,u}(m,n)$.
\end{thm}

Obviously~$K\cup \emptyset=K$ and~$\emptyset\cup L=L$ for every
languages~$K,L\subseteq \Sigma^*$.  Together with the
Theorems~\ref{thm:union-upper-range},~\ref{thm:union-lower-range},
and~\ref{thm:union-mid-range} we obtain the following corollary.

\begin{cor}\label{cor:union}
  We have
	$$
        g_{\cup}^{asc,u}(m,n)=g_{\cup,\pfa}^{asc,u}(m,n)=
        \begin{cases}
          \{n\} & \mbox{if $n= 0$,}\\
          \{m\} & \mbox{if $m= 0$,}\\
          \N & \mbox{otherwise.}
        \end{cases}
		$$
              \end{cor}

              \subsection{Difference}

              Now let us come to the difference operation which was
              also considered in~\cite{Da16}.  For deterministic
              finite automata with no restrictions the following
              result was shown:
$$
g^{\asc,u}_{\setminus}(m,n)=g^{\asc}_{\setminus}(m,n)=
\begin{cases}
  \{m\} & \mbox{if $n=0$}\\
  \{0\} & \mbox{if $m=0$}\\
  \{0\}\cup\N & \mbox{otherwise,}
\end{cases}
$$
Again the languages that prove these results are \emph{not} accepted
by any \pfa.

This subsection is structured as follows. First we will use the fact
that~$K\setminus L=K\cap C(L)$ for all finite languages~$K$ and~$L$ to
show that all numbers in the range~$[0,m]$ are reachable and for
all~$n,m\geq 1$ with~$\alpha\bmod m=0$ the numbers~$\alpha$ are
reachable as well. The previously mentioned fact for the set
difference of two languages allows us to prove the first two
statements by constructing~$A$ and~$B$ such that the minimal \dfa\
accepting the language of the direct product of~$A$ and~$\overline{B}$
has the required size, for~$\overline{B}$ being equal to~$B$ except
that its set of accepting states is complemented.  Afterwards we prove
that for all~$n,m\geq 1$ with~$\alpha\bmod m>0$ the numbers~$\alpha$
are reachable, too.

\begin{lem}\label{lem:diff-lower-range}
  For~$m,n\geq 1$ we have $[0,m]\subset
  g_{\setminus,\pfa}^{asc,u}(m,n)$.
\end{lem}

\begin{proof}
  Let~$A=A_w$ and~$B=A_{w'}$
  for~$w=(0^n1)^\alpha(10^n)^{m-\alpha}0^{n+1}$ and~$w'=1^n0^1$. That
  means~$L(A)\setminus L(B)$ is accepted by~$A_{w''}$
  for~$w''=(0^n1)^\alpha0^{(n+1)\cdot (m-\alpha+1)}$. We leave it to
  the reader to observe that the three involved automata are minimal.
  The \pfa~$A_{w''}$ has accepting state complexity~$\alpha$ which
  proves the statement of the lemma.
\end{proof}

The next lemma shows that the every number~$\alpha$ in the upper
range~$\N_{\geq m+1}$ is obtainable if~$\alpha$ fulfills~$\alpha\bmod
m=0$.

\begin{lem}\label{lem:diff-upper-range-less}
  For~$m,n\geq 1$ and~$\alpha\in\N_{\geq m+1}$ with~$\alpha\bmod m=0$
  we have~$\alpha\in g_{\setminus,\pfa}^{asc,u}(m,n).$
\end{lem}

\begin{proof}
  Let~$\alpha=mx+(\alpha\bmod m)$ and we set~$A=A_w$ and~$B=A_{w'}$
  for~$w=1^m0^k$ and~$w'=0^x1^n$, where~$k$ is the smallest positive
  integer such that $\gcd(m+k,x+n)=1$. It is not hard to see that~$A$
  and~$B$ are minimal and that the cross product \dfa\ of~$A$
  and~$\overline{B}=A_{w''}$, for~$w''=1^x0^n$, has~$xm$ accepting
  states and is minimal, too.
\end{proof}

The next lemma shows that the upper interval~$\N_{\geq m+1}$ is also
attainable for~$\alpha\bmod m\neq 0$ which clearly proves the range to
be reachable for all numbers in the interval.

\begin{lem}\label{lem:diff-upper-range-n-1}
  For~$m,n\geq 1$ and~$\alpha\in\N_{\geq m+1}$ with~$\alpha\bmod
  m\neq0$ we have~$\alpha\in g_{\setminus,\pfa}^{asc,u}(m,n).$
\end{lem}

By combining the Lemmata~\ref{lem:diff-lower-range},
\ref{lem:diff-upper-range-less}, and~\ref{lem:diff-upper-range-n-1} we
obtain the following corollary.

\begin{cor}\label{cor:set-minus}
  We have
	$$
	g^{\asc,u}_{\setminus}(m,n)=g_{\setminus,\pfa}^{asc,u}(m,n)=
	\begin{cases}
          \{m\} & \mbox{if $n=0$,}\\
          \{0\} & \mbox{if $m=0$,}\\
          \{0\}\cup\N & \mbox{otherwise.}
	\end{cases}
	$$
\end{cor}

\begin{proof}
  Since~$K\setminus \emptyset=K$ and~$\emptyset\setminus L=\emptyset$
  for every languages~$K,L\subseteq \Sigma^*$ the first two statements
  follow immediately.  Additionally the last statement follows by the
  Lemmata~\ref{lem:diff-lower-range}, \ref{lem:diff-upper-range-less},
  and~\ref{lem:diff-upper-range-n-1}.
\end{proof}

\subsection{Intersection}

We show that the left open unary case for the intersection operation
for both \pfa s and \dfa s differs from the solved general case. It is
not hard to see that at most the numbers in the range~$[1,nm]$ can be
reached. We split this interval into three smaller ones,
namely~$[0,\max\{n,m\}],$ $[\max\{n,m\}+1,nm-\min\{n,m\}]$
and~$[nm-\min\{n,m\}+1,nm]$.

\begin{lem}\label{lem:intersection-lower-interval}
  We have~$[0,\max\{n,m\}]\subseteq g^{\asc,u}_{\cap,\pfa}(m,n)$.
\end{lem}

\begin{proof}
  Since for the intersection operation the ordering of the input
  languages is irrelevant we assume~$m\geq n$.  Let~$A=A_w$
  and~$B=A_{w'}$ for~$w=(10^n)^\alpha(0^n1)^{m-\alpha}0^{n+1}$
  and~$w'=1^n0$, \resp.  Both \pfa s are minimal since the last~$n+1$
  states of~$A$ do not contain an accepting state and the accepting
  states of~$B$ form a sequence.  The minimal \dfa\ accepting the
  language~$L(A)\cap L(B)$ is the \pfa~$A_{w''}$
  for~$w''=(10^n)^\alpha0^{(n+1)\cdot(m-\alpha+1)}$ which obviously
  has accepting state complexity~$\alpha$.
\end{proof}

We found out by exhaustive search that for small~$m$ and~$n$ the
following conjecture holds.

\begin{con}
  All numbers in~$[\max\{n,m\}+1,nm-\min\{n,m\}]$ which are \emph{not}
  in
  \begin{multline*}
    [\max\{n,m\},n+m]\cup\\
    \hspace{0.5cm}\{\,t_n x_m\mid t_n\mbox{ is a nonzero divisor of~$n$ and } 0\leq x_m\leq (nm-\min\{n,m\})\div t_n\,\}\,\cup\\
    \hspace{1cm}\{\,t_m x_n\mid t_m\mbox{ is a nonzero divisor of~$m$
      and } 0\leq x_n\leq (nm-\min\{n,m\})\div t_m\,\}
  \end{multline*}
  are magic.
\end{con}

Next we investigate the numbers in the
range~$[nm-\min\{n,m\}+1,nm]$. For showing that all numbers
except~$nm$ are magic we prove the following structural property of
the cross product for \pfa s.

\begin{lem}\label{lem:intersection-reachability}
  Let~$q_0,q_1$ and~$p_0,p_1$ be states of the minimal unary \pfa
  s~$A$ and~$B$, \resp. If~$(q_0,p_0),$ $(q_1,p_0)$ and~$(q_0,p_1)$
  are initially reachable in the cross product automaton~$C$
  then~$(q_1,p_1)$ is initially reachable, too.
\end{lem}

\begin{proof}
  Let~$\Sigma=\{a\}$ be the input alphabet of~$A,$ $B$,
  and~$C$. Since~$(q_0,p_0),$ $(q_1,p_0)$, and~$(q_0,p_1)$ are
  initially reachable in~$C$ we know that there are words~$w_q$
  and~$w_p$ which map~$(q_0,p_0),$ onto~$(q_1,p_0)$ and~$(q_0,p_0),$
  onto~$(q_0,p_1)$. Because~$A$ and~$B$ are unary we observe
  that~$w_q$ and~$w_p$ induce the identity on~$B$ and~$A$, \resp. This
  implies that~$(q_0,p_0)\cdot w_qw_p=(q_1,p_0)\cdot w_p=(q_1,p_1)$
  which proves the stated claim.
\end{proof}

One may ask whether Lemma~\ref{lem:intersection-reachability} holds
for alphabets of arbitrary size. In general it is not true that~$w_q$
and~$w_p$ induce the identity on~$B$ and~$A$, \resp. Instead those
words induce a cycle on~$B$ and~$A$ that has a size that divides the
order of the word. If the cycle of~$A$ and~$B$ contain~$q_0,$ $q_1$
and~$p_0,$ $p_1$, \resp, the statement of the Lemma remains true. We
leave it to the reader to prove or disprove the lemma above for
alphabets of at least two letters.  As mentioned before we use
Lemma~\ref{lem:intersection-reachability} to prove that a whole range
of numbers in the upper interval cannot be reached.

\begin{thm}\label{thm:unary-intersection-upper-interval}
  We have~$[nm-\min\{n,m\}+1,nm-1]\nsubseteq
  g_{\cap,\pfa}^{\asc,u}(m,n)$.
\end{thm}

\begin{proof}
  Let~$\alpha\in [nm-\min\{n,m\}+1,nm-1]$.  Clearly
  Lemma~\ref{lem:intersection-reachability} implies that for all \pfa
  s~$A$ and~$B$ their cross product automaton has either less
  than~$nm-\min\{n,m\}+1$ or~$nm$ initially reachable accepting
  states. On the other hand a result from~\cite[Lemma 4]{HoRa21a}
  implies that if a \pfa\ has~$nm$ accepting states and it is not
  minimal then its minimal \dfa\ has~$t$ accepting states for a
  divisor of~$nm$. Since every divisor of~$nm$ is less
  than~$nm-\min\{n,m\}+1$ the claim of the theorem follows.
\end{proof}

If we look at the cross product automaton~$C$ of unary \dfa s~$A$
and~$B$ we see that for every state~$q$ that is in the tail of~$A$
or~$B$ there is exactly one initially reachable state in~$C$ which
contains~$q$ as one of its components. So we obtain that~$C$ contains
at most~$nm-\min\{n,m\}+1$ initially reachable accepting states.
Together with Theorem~\ref{thm:unary-intersection-upper-interval} we
obtain the following corollary.

\begin{cor}\label{cor:unary-intersection-dfa}
  We have~$[nm-\min\{n,m\}+2,nm-1]\nsubseteq g_{\cap }^{\asc,u}(m,n)$.
\end{cor}

As mentioned before the upper bound for~$g_{\cap,\pfa}^{\asc,u}(m,n)$
is not a magic number which is proven in the following lemma.

\begin{lem}
  \label{lem:unary-intersection-nm}
  We have~$nm\in g_{\cap,\pfa}^{\asc,u}(m,n)$.
\end{lem}

\begin{proof}
  Let~$A=A_w$ and~$B=A_{w'}$ for~$w=1^m0^n$ and~$1^n0^{m+1}$,
  \resp. Since~$n+m$ and~$n+m+1$ are coprime it is obvious that length
  of their product automaton~$C$ is~$(n+m)\cdot(n+m+1)$.  So each pair
  of accepting states is initially reachable.  We observe that there
  are~$\max\{n,m\}-\min\{n,m\}$ sequences of accepting states of
  length~$\min\{n,m\}$. All of those sequences follow each other,
  i.e., only non-accepting states are between them.  There are also
  shorter sequences of accepting states in~$C$, e.g., an accepting
  states which follows and is followed by a non-accepting state. This
  implies that~$C$ has to be minimal which proves the stated claim.
\end{proof}

\subsection{Reversal}

The results of this subsection are in contrast to the general case
where arbitrary \dfa s are considered. Here the restriction for the
input automaton to be a \pfa\ provides magic numbers which are not
magic if the input automaton is not restricted.  For deterministic
finite automata with no restrictions the following result was proven
in~\cite{HoHo18}:
$$
g^{\asc}_{R}(m,n)=
\begin{cases}
  \{0\} & \mbox{if $m=0$,}\\
  \N & \mbox{otherwise.}
\end{cases}
$$
We show that in the case of permutation automata the number~$\alpha=1$
is magic for all~$m\geq 2$. Before we do this we need a special \pfa\
that plays an important role for the reversal operation. We want to
mention that for a unary language~$L$ its reversal~$L^R$ is equal
to~$L$. So we will only consider languages with at least two different
letters.  First we define~$\binom{S}{k}$ for a finite set~$S$ and a
non-negative integer to be the set of all subsets of~$S$ which have
size~$k$.  For a \pfa~$A=(\{q_0,q_1,\ldots,
q_{k-1},\},\Sigma,{}\cdot_A{},q_0,F_A)$, we
define \[A_R=(\binom{Q_A}{|F_A|},\Sigma,{}\cdot_{A_R}{},F_A,\{\,R\in
\binom{Q_A}{|F_A|}\mid q_0\in R\,\}),\] where $R\cdot_{A_R}
w=\{\,q\cdot_A w^{-1}\mid q\in R,w\in\Sigma\,\}$ for
all~$R\in\binom{Q_A}{|F_A|}$.  We want to mention here that~$A_R$ is a
well-defined \dfa\ since for every word~$w$ the mapping~$w^{-1}$ is
uniquely defined because~$A$ is a \pfa. Because~$w^{-1}$ applies the
reverse transitions to every state of~$A$ and every state of~$A_R$
that contains the initial state of~$A$ is accepting so it is not hard
to see that~$A_R$ accepts the language~$L(A)^R$. Before we prove our
results for the accepting state complexity of the reversal operation
of \pfa s we derive two structural properties of the \dfa~$A_R$. First
we count the number of initially reachable states in~$A_R$.

\begin{lem}\label{lem:number-AR}
  Let~$A$ be a minimal \pfa. Then there is an integer~$x\geq 1$ such
  that for every state~$q$ of~$A$ there are~$x$ initially reachable
  states in~$A_R$ containing~$q$.
\end{lem}

\begin{proof}
  Let~$q$ be an arbitrary state of~$A$. Assume there are~$x\geq 1$
  states~$R_0$, $R_1,\ldots,R_{x-1}$ in~$A_R$ which contain~$q$.
  Since~$A$ is an~\pfa\ the images of those states are different
  regardless of the choice of the mapping. If we apply the mapping
  which maps~$q$ onto~$q'$ for any other state~$q'$ of~$A$ it follows
  directly that there are at least~$x$ states of~$A_R$ which
  contain~$q'$. Since this argument can be used symmetrically the
  claim of the lemma follows.
\end{proof}

Since the states of~$A_R$ are in turn sets we prove the following
property of~$A_R$ which is the automata theoretical interpretation of
the fact that bijections on elements induce bijections on sets of
those elements.

\begin{lem}
  \label{lem:AR-pfa-property}
  For every \pfa~$A$ the \dfa~$A_R$ is a \pfa, too.
\end{lem}

\begin{proof}
  Since~$A$ is a \pfa\ for every letter~$a$ the preimage of any
  state~$q$ of~$A$ is uniquely defined. By applying this property to
  every state~$q$ in a state~$R$ of~$A_R$ we directly obtain the
  unique preimage of~$R$.
\end{proof}

Now we will prove our magic number result for the \pfa\ case of the
reversal operation.

\begin{lem}
  \label{lem:reversal-number-one}
  Let $m\geq 2$. Then there exists \emph{no} \pfa~$A$ with $asc(A)=m$
  such that~$\asc(A_R)=1$.
\end{lem}

The result of the previous lemma proves the inequality statement of
our main theorem for the accepting state complexity of the reversal
operation of p-regular languages. Obviously we also prove that for~$m$
equal two every number unequal one is not magic.  We do this by
constructing an automaton~$A$ such that~$A_R$ has~$\alpha\cdot k\div
m=\binom{k}{m}$ initially reachable states while every state of~$A$
appears in exactly~$\alpha$ of them.

\begin{thm}\label{thm:reversal}
  We have
		$$g_{R,\pfa}^{\asc}(m)=
		\begin{cases}
                  \{0\} & \mbox{if $m=0$,}\\
                  \{1\} & \mbox{if $m=1$,}\\
                  \N_{\geq 2} & \mbox{if $m=2$,}
		\end{cases}
		$$
		and $ g_{R,\pfa}^{\asc}(m)\neq \N$ if $m\geq 3$.
                Therefore~$g_{R,\pfa}^{\asc}(m)\neq g_{R}^{\asc}(m)$.
\end{thm}

We note here that for~$m\geq 3$ the equation~$\alpha\cdot k\div
m=\binom{k}{m}$ has no integer solution for many values of~$\alpha$
and~$m$ which can be easily confirmed. For those values of $\alpha$
and~$m$ for which the equation has an integer solution we
obtain~$\alpha\in g_{R,\pfa}^{\asc}(m)$ in similiar fashion like
for~$m=2$. Nevertheless we conjecture the following:

\begin{con}\label{con:reversal}
  We have
$$g_{R,\pfa}^{\asc}(m)=
\begin{cases}
  \{0\} & \mbox{if $m=0$,}\\
  \{1\} & \mbox{if $m=1$,}\\
  \N_{\geq 2} & \mbox{if $m\geq2$.}
\end{cases}
$$
\end{con}

Clearly this would mean that~$\alpha=1$ is the only number which is
magic for the reversal of p-regular languages and non-magic for
arbitrary regular languages.

\subsection{Quotient}

For two \dfa s~$A$ and~$B$ the right quotient~$L(A)L(B)^{-1}$ can be
accepted by the \dfa~$\tilde{A}$ which can be obtained from~$A$ by
exchanging its set of accepting states~$F$ by~$\{\,q\mid q\cdot w \in
F\mbox{ for some $w$ in $L(B)$}\,\}$, which we denote
by~$\tilde{F}$. It is obvious that~$\tilde{A}$ is a \pfa\ if~$A$ is a
\pfa.  Additionally, if $s$ is the initial state of $A$, then the
automaton obtained from $A$ by making all states in $\{\,s\cdot w\mid
w\in L\,\}$ initial accepts~$L(B)^{-1}L(A)$.  Since for unary
languages the left and right quotient coincide no distinction is made
at this point and we use the right quotient unless otherwise stated.
For regular languages in general the following was shown
in~\cite{HoHo18}:
$$
g^{\asc,u}_{^{-1}}(m,n)=g^{\asc}_{^{-1}}(m,n)=
\begin{cases}
  \{0\} & \mbox{if $m=0$ or~$n=0$,}\\
  \{0\}\cup\N & \mbox{otherwise.}
\end{cases}
$$

Clearly the first statement follows directly from the fact
that~$K\emptyset^{-1}=\emptyset L^{-1}=\emptyset$ for all
languages~$K$ and~$L$.  But we show that last statement does not hold
for the class of p-regular languages.  For this we distinguish
whether~$n$ is equal one or at least equal to two.  First we show
which numbers are reachable if~$n$ equals one.

\begin{lem}\label{lem:quotient-n-equal-one-construct}
  We have~$[1,m]\subseteq g_{^{-1},\pfa}^{\asc,u}(m,1)$.
\end{lem}

\begin{proof}
  Let~$\alpha$ be in $[1,m]$. Define~$A=A_w$ and~$B=A_{w'}$
  for~$w=1^\alpha0^{m+1-\alpha}(10^{m})^{m-\alpha}0^{m+1}$
  and~$w'=010^{m-1}$. We observe that~$A$ has
  length~$(m+1)(m-\alpha+2)$ and~$B$ has length~$m+1$. It is not hard
  to see that
  \[
  L(A)=\{\,a^{i+x(m+1)(m-\alpha+2)}\mid 0\leq i\leq\alpha-1,0\le
  x\,\}\cup\{\,a^{(m+1)i+x(m+1)(m-\alpha+2)}\mid 1\leq i\leq
  m-\alpha,0\le x\,\}\] and
  \[L(B)=\{\,a^{(m+1)i+1}\mid i\in\N\cup\{0\}\,\},\] for~$a$ being the
  letter of the input alphabet of~$A$ and~$B$. We observe that the
  \pfa~$\tilde{A}$ has the set of accepting states~$\tilde{F}$ that
  contains exactly the elements $q_{(i-((m+1)j+1))\bmod
    (m+1)(m-\alpha+2)}$, for $0\leq i\leq \alpha-1$ or
  $i\in\{\,(m+1)\ell\mid 1\leq \ell\leq m-\alpha\,\}$ and
  $j\in\N\cup\{0\}$. Alternatively we can write
  \[ \tilde{F}=\{\, q_{(i-((m+1)j+1))\bmod (m+1)(m-\alpha+2)}\mid
  0\leq i\leq \alpha-1, j\in\N\cup\{0\} \,\} \] because
  \[\displaylines{\quad 0-((m+1)(m-\alpha+2-j)+1) \bmod
    (m+1)(m-\alpha+2)\hfill\cr \hfill{}= (m+1)j-((m+1)\cdot 0)+1 \bmod
    (m+1)(m-\alpha+2)\quad\cr}\]
  holds.  One may observe that~$\tilde{F}$ contains all accepting
  states of~$A$ but their index is decreased by one modulo the length
  of~$A$. If we shift those states again by an arbitrary multiple
  of~$(m+1)$ we obtain the remaining states in~$\tilde{F}$.
  Clearly~$\tilde{F}$ does not contain other
  states. Therefore~$\tilde{A}=A_{w''}$,
  for~$w''=1^{\alpha-1}0^{m+1-\alpha}(1^{\alpha}0^{m+1-\alpha})^{m-\alpha+1}1$. Indeed
  this \pfa\ is not minimal, i.e., all of its sequences
  $1^{\alpha-1}0^{m+1-\alpha}1$ are equivalent. Thus the minimal \pfa\
  accepting the language~$L(A)L(B)^{-1}$ is~$A_{w'''}$, for the
  word~$w'''=1^{\alpha-1}0^{m+1-\alpha}1$ which has accepting state
  complexity~$\alpha$.
\end{proof}

Next we prove that every number which is not reachable due to the
previous lemma is magic.

\begin{lem}\label{lem:quotient-n-equal-one-magic}
  We have~$[1,m]= g_{^{-1},\pfa}^{\asc,u}(m,1)$.
\end{lem}

\begin{proof}
  Due to the proof of Lemma~\ref{lem:quotient-n-equal-one-construct}
  it remains to show that $\N_{\geq m+1}$ is not
  in~$g_{^{-1},\pfa}^{\asc,u}(m,1)$.  Therefore let~$A$ and~$B$ be
  unary minimal \pfa s with~$m$ and~$n$ accepting states, \resp.
  Recall that~$\tilde{A}$ is the \pfa\ obtained from~$A$ by replacing
  its set of accepting states with~$\tilde{F}=\{\,q\mid q\cdot w \in
  F\mbox{ for some $w$ in $L(B)$}\,\}$.  We observe that the set of
  accepting states of~$\tilde{A}$ is equal
  \[\tilde{F}=\{\, q_{(i-(jk'+\ell))\bmod k}\mid i\in I_A,
  j\in\N\cup\{0\} \,\}\] for~$I_A$ being the index set of the
  accepting states of~$A$, $q_{\ell}$ being the accepting state of~$B$
  and~$k, k'$ being the number of states of~$A$ and~$B$, \resp.  For
  an arbitrary but fixed~$i\in I_A$ we see that each of the states
  \[
  q_{(i-(0k'+\ell))\bmod k},q_{(i-(1k'+\ell))\bmod k},
  q_{(i-(2k'+\ell))\bmod k},\ldots
  \]
  can be mapped by~$a^{k'}$ onto its predecessor.  Since this holds
  for every~$i\in I_A$ those states have to be equivalent which proves
  that~$\tilde{A}$ contains at most~$m$ inequivalent accepting states.
\end{proof}

Now we generalize Lemma~\ref{lem:quotient-n-equal-one-construct},
for~$n\geq 2$.

\begin{lem}\label{lem:quotient-n-equal-one-construct-geq2}
  We have~$[1,mn]\subseteq g_{^{-1},\pfa}^{\asc,u}(m,n)$, for~$n\geq
  2$.
\end{lem}

Next we rule out every number that is not reachable by
Lemma~\ref{lem:quotient-n-equal-one-construct-geq2}. Like in the proof
of Lemma~\ref{lem:quotient-n-equal-one-magic} one observes that the
set of accepting states of the \dfa\ accepting the quotient language
of two p-regular languages~$L_1$ and~$L_2$ of accepting state
complexity~$m$ and~$n$, \resp, is given by applying the following two
steps. First the accepting states of the minimal \dfa\ accepting~$L_1$
are shifted onto~$n$ positions. Afterwards these~$mn$ accepting states
are cyclic replicated by the length of the \dfa\
accepting~$L_1$. Since the \dfa\ accepting the quotient language is a
\pfa\ all cyclic replications are equivalent.

\begin{lem}\label{lem:quotient-n-equal-one-magic-geq2}
  We have~$[1,mn]= g_{^{-1},\pfa}^{\asc,u}(m,n)$.
\end{lem}

By using the previous four
Lemmata~\ref{lem:quotient-n-equal-one-construct},
\ref{lem:quotient-n-equal-one-magic},
\ref{lem:quotient-n-equal-one-construct-geq2},
and~\ref{lem:quotient-n-equal-one-magic-geq2} we deduce the following
corollary.

\begin{cor}\label{cor:quotient}
  We have
	$$
	g^{\asc,u}_{^{-1},\pfa}(m,n)=
	\begin{cases}
          \{0\} & \mbox{if $m=0$ or~$n=0$,}\\
          [1,mn] & \mbox{otherwise.}
	\end{cases}
	$$
	Therefore~$g^{\asc,u}_{^{-1},\pfa}(m,n)\neq
        g^{\asc,u}_{^{-1}}(m,n)$.
\end{cor}

The accepting state complexity for the quotient operation on languages
accepted by permutation automata with larger input alphabets has to be
left open and is subject to further research.

% \nocite{*}
%\bibliographystyle{eptcs} 
%\bibliography{shortmarkus,markus}

\def\lastmodified{23.12.2008}\def\id#1{#1}

\end{document}